\documentclass[journal,12pt,draftclsnofoot,onecolumn]{IEEEtran}
\usepackage{amsmath,epsfig,amssymb,mathrsfs}
\usepackage{cite}
\usepackage{times}
\usepackage{graphics}
\usepackage{array}
\usepackage[figuresright]{rotating}
\usepackage{subfigure}
\usepackage{rotating}
\usepackage{multirow}
\usepackage{wrapfig}
\usepackage{booktabs}
\usepackage{array}
\usepackage{comment}
\usepackage{graphicx}
\usepackage{float}
\usepackage{amssymb}
\usepackage{amsthm}
\usepackage{lipsum} 
\usepackage{algorithmic}
\usepackage{algorithm}
\usepackage{fancybox}
\usepackage{cite}
\usepackage{graphicx}
\usepackage{dsfont}
\usepackage{balance}
\usepackage{epstopdf}

\newtheorem{thm}{Theorem}

\newtheorem{lem}[thm]{Lemma}

\floatname{algorithm}{Algorithm}

\usepackage{colortbl}

\begin{document}

\title{Elastic-Net: Boosting Energy Efficiency and Resource Utilization in 5G C-RANs}
\author{{\bf Abolfazl Hajisami, Tuyen X. Tran, and Dario Pompili}\\
Department of Electrical and Computer Engineering \\Rutgers University--New Brunswick, NJ, USA\\
e-mail: \{hajisamik, tuyen.tran, pompili\}@cac.rutgers.edu}

\maketitle
\thispagestyle{empty}

\begin{abstract}
Current Distributed Radio Access Networks~(D-RANs), which are characterized by a static configuration and deployment of Base Stations~(BSs), have exposed their limitations in handling the temporal and geographical fluctuations of capacity demands. At the same time, each BS's spectrum and computing resources are only used by the active users in the cell range, causing idle BSs in some areas/times and overloaded BSs in other areas/times. Recently, Cloud Radio Access Network~(C-RAN) has been introduced as a new centralized paradigm for wireless cellular networks in which---through virtualization---the BSs are physically decoupled into Virtual Base Stations~(VBSs) and Remote Radio Heads~(RRHs). In this paper, a novel elastic framework aimed at fully exploiting the potential of C-RAN is proposed, which is able to adapt to the fluctuation in capacity demand while at the same time maximizing the energy efficiency and resource utilization. Simulation and testbed experiment results are presented to illustrate the performance gains of the proposed elastic solution against the current static deployment.
\end{abstract}
\begin{IEEEkeywords}
Cloud Radio Access Network; Virtual Base Station; Energy Efficiency; Resource Utilization.
\end{IEEEkeywords}

\section{Introduction}\label{sec:introduction}
\textbf{Motivation:}
Over the last few years, proliferation of personal mobile computing devices like tablets and smartphones along with a plethora of data-intensive mobile applications has resulted in a tremendous increase in demand for ubiquitous and high data rate wireless communications. The current practice to enhance the spectral efficiency and data rate is to increase the number of Base Stations~(BSs) and go for smaller cells so as to increase the band reuse factor. \emph{However, it is studied that increasing the BS density or the number of transmit antennas will decrease the energy efficiency due to the dynamic traffic variation}~\cite{li2013throughput}. The number of active users at different locations varies depending on the time of the day and week. This movement of mobile network load based on time of the day/week is referred to as the ``tidal effect".  In the traditional Distributed Radio Access Network~(D-RAN), \emph{each BS's spectrum and computing resources are only used by the active users in the cell range}. Hence, deploying small cells for the peak (worst case) traffic time leads to grossly under-utilized BSs in some areas/times and is highly energy inefficient, while deploying cell for the average traffic time leads to oversubscribed BSs in some other areas/times. 

At the stage of network planning, cell size and capacity are usually determined based on the estimation of peak traffic load. However, \emph{due to the tidal effect, there are no fixed cell size and transmission power that optimize the overall power consumption of the network}. This means that the use of small cells is quite efficient in terms of power consumption as well as utilization of spectrum and computing resources when the capacity demand is high and evenly distributed in space; however, it becomes less so when the data traffic is low and/or uneven due to the static resource provisioning and fixed BS power consumptions. On the other hand, the economic impact of power consumption is particularly dire in emerging markets and the Fifth Generation~(5G) of wireless networks must be not only spectral efficient but also energy efficient (e.g., a 1000$\times$ improvement in energy efficiency is expected by 2020).  Although several recent efforts have been made to increase spectral efficiency~\cite{hajimirsadeghi2016inter} and to reduce the power consumption of existing small cell networks, limited attention has been given towards optimizing the overall network deployment. Hence, a novel design and architecture is necessary for the next generation of wireless network to overcome these challenges.

\textbf{A New Centralized Cellular Network Paradigm:}
Cloud Radio Access Network~(C-RAN)~\cite{whitepaper} is a new architecture for the next generation of wireless cellular networks that allows for a dynamic reconfiguration of spectrum/computing resources (Fig.~\ref{fig:cran_example}). C-RAN consists of three parts: 1)~Remote Radio Heads~(RRHs) plus antennae, which are located at the remote site and are controlled by Virtual Base Stations~(VBSs) housed in a centralized processing pool, 2)~the Base Band Unit~(BBU) (known as VBS pool) composed of high-speed programmable processors and real-time virtualization technology to carry out the digital processing tasks, and 3)~low-latency high-bandwidth optical fibers, which connect the RRHs to the VBS pool. The communication functionalities of the VBSs are implemented on Virtual Machines~(VMs), which are housed in one or more racks of a small cloud datacenter. This centralized characteristic along with virtualization technology and low-cost relay-like RRHs provides a higher degree of freedom to make optimized decisions, and has made C-RAN a promising candidate to be incorporated into the 5G networks~\cite{tran2016quaro, hajisami2017dynamic,hajisamicocktail,hajimirsadeghi2017joint}.



\textbf{Our Contributions:}
In this paper, we focus on optimizing the power consumption and resource utilization by leveraging the full potential of the C-RAN architecture. We propose a novel elastic resource provisioning framework, called ``Elastic-Net'', to minimize the power consumption while addressing the fluctuations in per-user capacity demand. In our solution, we divide the covered region of the network into clusters based on the traffic model and, within each cluster, we dynamically adapt the active RRH density, transmission power, and size of the VM\footnote{The size of a VM is represented in terms of its processing power, memory and storage capacity, and network interface speed.} based on the traffic fluctuations so as to minimize the power consumption while maximizing the resource utilization. To minimize the power consumption in the cell sites while ensuring a certain minimum coverage and data rate, we propose to dynamically optimize and adapt the RRH density and transmission power based on the traffic demand and user density. Likewise, to minimize the power consumption in the cloud we dynamically optimize and adapt the size of the VMs while ensuring that the frame-processing deadline is met.

\textbf{Paper Outline:}
In Sect.~\ref{sec:NM}, we present the system model. 
In Sect.~\ref{sec:elastic_provisioning}, we formulate the problem and describe our demand-aware provisioning framework. In Sect.~\ref{sec:simulation}, we validate our statements through simulation and testbed experiments. 
Finally, we conclude the paper in Sect.~\ref{sec:conclusion}.

\section{System Model}\label{sec:NM}
We consider a C-RAN downlink system and assume that each user is served by the nearest active RRH. The RRHs and users are distributed according to two independent Poisson Point Processes~(PPPs) in $\mathbb{R}^2$, denoted as $\Phi_r$ and $\Phi_u(t)$, respectively. Let $\lambda_{r}$ and $\lambda _u(t)$ denote the RRH density and the time-dependent user density, respectively. The set of all RRHs is denoted by $\mathcal{L}=\left\{ {1, \ldots ,L} \right\}$, where $\mathcal{A} \subseteq \mathcal{L}$ is the set of active RRHs and $\mathcal{Z} \subseteq \mathcal{L}$ is the set of inactive RRHs ($\mathcal{A} \cup \mathcal{Z} = \mathcal{L}$). Let also ${\mu _a(t)}$ $(0 \le {\mu _a}(t) \le 1)$ denote the RRH activity factor, which indicates the ratio of active RRHs to all RRHs, where $\lambda _{r}^a (t) = \mu _a (t)\lambda _r$ is the time-dependent density of active RRHs. The total bandwidth is denoted by $B$ and the bandwidth per user is given by $B_{u}(t) = {B}\frac{\lambda _{r}^a (t)}{{\lambda _u(t)}}$.

\begin{figure*}[th]
 \centering
 \resizebox{\textwidth}{!} {
 \begin{tabular}{cc}
\includegraphics[width=0.5\textwidth,height=2.5in]{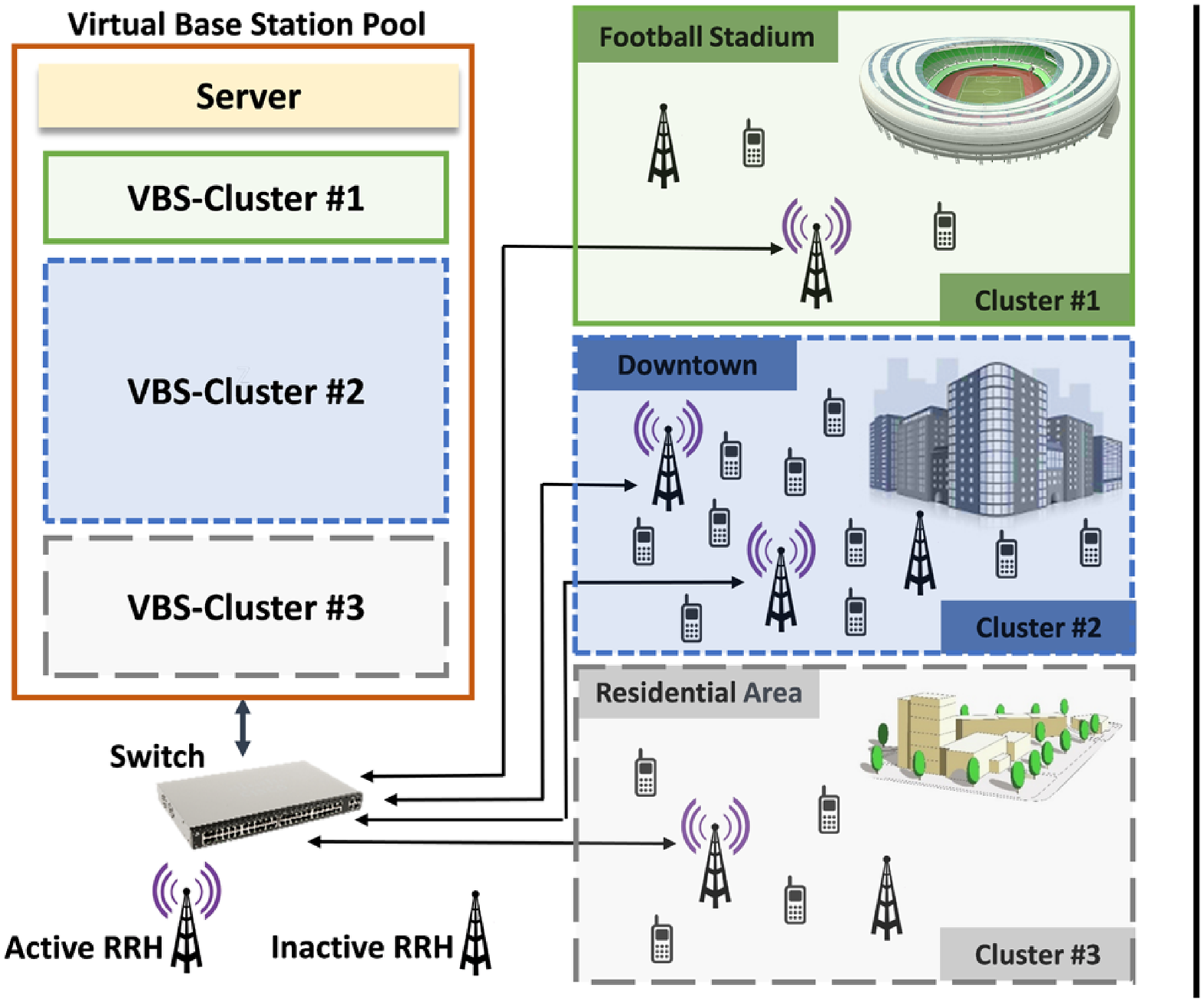} &
\hspace{-0.1cm}\includegraphics[width=0.48\textwidth,height=2.5in]{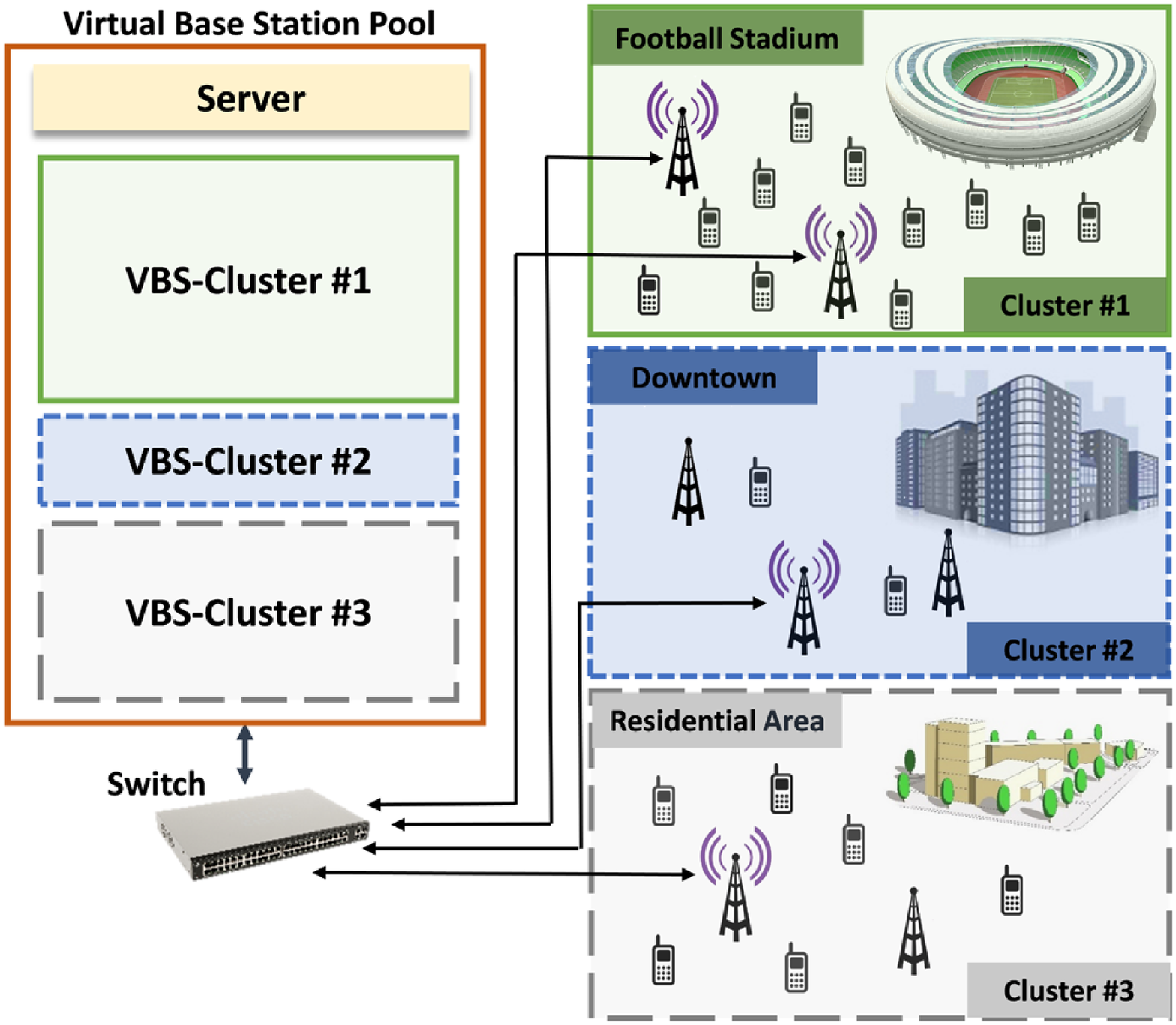} \\
\small (a) Day & \small(b) Night
\end{tabular}
}
\caption{The use of virtualization in C-RAN allows dynamic re-provisioning of spectrum and computing resources (visualized here using different sizes) to Virtual Base Stations~(VBSs) based on demand fluctuation; (a) and (b) illustrate the movement of mobile network load from the downtown office area to the residential and recreational areas over the course of a day and the corresponding changes in active RRH density and VBS size (we have used different icons for active RRH and inactive RRH). }\label{fig:cran_example}
\end{figure*}

\textbf{RRH and Transport Network Power Consumption Model:}
Since in C-RAN the BSs are decoupled into RRHs and VBSs, we divide the network power consumption into two parts: (i)~RRH and transport network power consumption and (ii)~VBS pool power consumption. For the power consumption of a RRH, we consider a linear model as in~\cite{auer2011much},
\begin{equation}\label{RRH_power}
{P_{rrh}} = \left\{ {\begin{array}{*{20}{cc}}
{P_{rrh}^a + \frac{1}{\eta }{P}}&{{\rm{if}}\hspace{0.15cm}{P} > 0}\\
{\hspace{-1.1cm}P_{rrh}^s}&{{\rm{if}}\hspace{0.15cm}{P} = 0}
\end{array}} \right.,
\end{equation}where $P_{rrh}^a$ is the active circuit power consumption, $\eta$ is the power amplifier efficiency, $P$ is the transmission power, and $P_{rrh}^s$ is the RRH power consumption in the sleep mode. We also consider the future Passive Optical Network~(PON) to provide low-cost, high-bandwidth, low-latency connections between the RRHs and VBS pool~\cite{dhaini2014energy}. PON comprises an Optical Line Terminal~(OLT) that resides in the VBS pool and connects a set of associated Optical Network Units~(ONUs) through a single fiber. In this paper, we consider \emph{fast/cyclic sleep} mode where the ONU state alternates between the active state (when the RRH is in the active state) and the sleep state (when the RRH is in the sleep state). Hence, the power consumption of the transport network is given as in~\cite{dhaini2014energy},
\begin{equation}\label{P_transport}
{P_{tn}} = {P_{olt}} + {P_{onu}},
\end{equation}where $P_{olt}$ is the OLT power consumption in the VBS pool and $P_{onu}$ is the ONU power consumption, given as,
\begin{equation}\label{P_onu}
{P_{onu}} = \left| \mathcal{A} \right|P_{tl}^a + \left| \mathcal{Z} \right|P_{tl}^s,
\end{equation}where $P_{tl}^a$ and $P_{tl}^s$ are the consumed power by each ONU  in the active and sleep mode, respectively. Since $P_{olt}$ is consumed in the VBS pool, we consider it in the power consumption of the VBS pool. Therefor, the overall area RRH and transport network power consumption is given by,
\begin{equation}\label{P_area}
\mathcal{P}_{area} = {\lambda _{r}^a(t)}({P_{rrh}^a} + \frac{1}{\eta }{P} + {P_{tl}^a}) + {\lambda _{r}^s(t)}({P_{rrh}^s} + {P_{tl}^s}).
\end{equation}

\textbf{VBS Pool Power Consumption Model:}
As discussed in~\cite{beloglazov2012energy}, compared to other system resources, the CPU consumes the main part of the power in the VBS pool; hence, in this work, we focus on minimizing its power consumption. In order to model the power consumption of the VBS pool, we introduce the notion of \emph{size of a VM}, which is represented in terms of its processing power~[CPU cycles per second]; for each VM, we consider the power model defined as,
\begin{equation}\label{VBS_Power}
P_{vm} = \Delta(t) P_{max} u(t) + \beta \Delta(t) P_{max} (1-u(t)) + P_{olt},
\end{equation}where $\Delta(t)$ is the size of the VM in terms of CPU cycles per second, $P_{max}$ is the maximum power consumed per unit VM size when the server is fully utilized, $\beta$ is the fraction of power consumed by the idle VM, and $u(t)$ is VM utilization. Note that, in our model, $\Delta(t)$ and $u(t)$ change over time due to the workload variation and hence they are functions of time.

\section{Elastic-Net: Demand-Aware Provisioning}\label{sec:elastic_provisioning}
In our solution, as depicted in Fig.~\ref{fig:cran_example}, we cluster the neighboring RRHs and their corresponding VBSs based on traffic model, and in each cluster we adapt the system parameters accordingly. We advocate \emph{demand-aware resource provisioning} where in each cluster the active RRH density, transmission power, and size of the VM are dynamically changed over time to minimize the power consumption and to meet the fluctuating traffic demand as well as network constraints. For instance, as shown in Fig.~\ref{fig:cran_example}(a)-(b), due to the higher capacity demand during day time in cluster~\#2~(a), we provision it with more active RRH and higher size of VBS-Cluster rather than during night time~(b), when we have less capacity demand. Hence, our objective is to obtain the optimal active RRH density, transmission power, and size of VM for each cluster so that the power consumption is minimized while meeting a predefined coverage probability, per-user data rate, and frame-processing time. Our optimization problem for the $i^{\rm {th}}$ cluster can be cast as,
\begin{subequations}
\label{first_opt}
\begin{align}
& \mathfrak{p}: \underset{{\mu_a,{P},\Delta}}{\text{argmin}}
& & {\mathcal{P}_{area}^i}({\mu_a (t,i),{P}(t,i)})+ P_{vm}^i(\Delta (t,i))\label{first_opt1} \\
& \text{subject to}
& & {{P_{{\mathop{\rm cov}} }} \ge \varepsilon P_{{\mathop{\rm cov}} }^{\infty}}, \label{first_opt2} \\
&&& {{R_u} \ge  {R_{\rm 0}}}, \label{first_opt3}\\
&&& {{T_{\rm dl}} \ge  {T_{\rm fr}}},
\end{align}
\end{subequations}where ${\mathcal{P}_{area}^i}({\mu_a (t,i),{P}(t,i)})$ and $P_{vm}^i(\Delta (t,i))$ are the area power consumption and VBS-Cluster power consumption of the $i^{th}$ cluster, respectively, $P_{{\rm{cov}}}^{{\infty}}$ is the coverage probability at no noise regime, $R_{\rm 0}$ is the per-user minimum data rate, $T_{\rm fr}$ is the frame-processing time, $T_{\rm dl}$ is frame deadline, and $\varepsilon$  is a positive number in $[0,1]$. $\mu_a(t,i)$, $P(t,i)$, and $\Delta (t,i)$ are also the RRH activity factor, transmission power, and size of the VM for the $i^{\rm {th}}$ cluster, respectively. Due to the temporal variation of traffic demand in each cluster, $\mu_a(t,i)$, $P(t,i)$, and $\Delta (t,i)$ are time dependent; hence, the optimal solution $[\mu_a^*(t,i),P^*(t,i),\Delta^* (t,i)]$ in general varies over time.

The density of active and inactive RRHs in the $i^{th}$ are,
\begin{subequations}
\label{active_inactive}
\begin{align}
&\lambda _r^a(t,i) = {\mu _a}(t,i){\lambda _r}(i), \label{eq:active}\\
&\lambda _r^s(t,i) = \left( {1 - {\mu _a}(t,i)} \right){\lambda _r}(i), \label{eq:inactive}
\end{align}
\end{subequations}where ${\lambda _r}(i)$ is the density of all RRHs in the $i^{th}$ cluster. By substituting~\eqref{eq:active} and \eqref{eq:inactive} into \eqref{P_area}, we can write,
\begin{equation}\label{P_I}
{\mathcal{P}_i}\left( {P(t,i),{\mu _a}(t,i)} \right) = {\lambda _r}(i)\left( {{\mu _a}(t,i){Q_1}(t,i) + {Q_2}} \right),
\end{equation}where
\begin{subequations}
\label{q1_a2}
\begin{align}
&{Q_1}(t,i) = P_{rrh}^a + \frac{1}{\eta }P(t,i) + P_{tl}^a - P_{rrh}^s - P_{tl}^s, \label{eq:q1}\\
&{Q_2} = P_{rrh}^s + P_{tl}^s. \label{eq:q2}
\end{align}
\end{subequations}

From~\eqref{P_I} and~\eqref{eq:q1}, we see that the objective function is non-convex because of the multiplication term of $\mu _a(t,i)$ and $P(t,i)$. To minimize the objective function, as in~\cite{lin2016multi}, we can use the coordinate descent algorithm and minimize $\mu_a(t,i)$, $P(t,i)$, and $\Delta (t,i)$ independently.

\subsection{Optimal Active RRH Density}
\begin{lem}\label{lemma_0}
The minimum RRH activity factor for which the constraint ${{R_u} \ge  {R_{\rm 0}}}$ is met is given by,
\begin{equation}\label{mu_a0}
\mu _a^*(t,i) = \frac{{{R_{\rm 0}}  \lambda _u \left( {t,i} \right)}}{{B\lambda _r \left( i \right)\left[ {\log _2 (1 + \gamma ) + \gamma ^{\frac{2}{\alpha }} \mathcal{A}\left( {\alpha ,\gamma } \right)} \right]}},
\end{equation}
where
\begin{equation}\label{L0}
\mathcal{A}\left( {\alpha ,\gamma } \right) = \int_\gamma ^\infty  {\frac{{x^{ - 2/\alpha } }}{{1 + x}}dx}.
\end{equation}
\end{lem}

\begin{proof}
The spectral efficiency achievable by a randomly chosen user when it is in coverage is given as in~\cite{dhillon2012modeling},
\begin{equation}\label{K_tier}
\tau (\alpha ,\gamma ) = \log _2 \left( {1 + \gamma } \right) + \gamma ^{\frac{2}{\alpha }} \mathcal{A}\left( {\alpha ,\gamma } \right).
\end{equation}Hence, the per-user data rate in the $i^{\rm th}$ cluster and at time instant $t$ can be written as,
\begin{equation}\label{data_rate}
R_u \left( {t,i} \right) = \frac{{B\mu _a \left( {t,i} \right)\lambda _r \left( i \right)}}{{\lambda _u \left( {t,i} \right)}}\left[ {\log _2 \left( {1 + \gamma } \right) + \gamma ^{\frac{2}{\alpha }} \mathcal{A}\left( {\alpha ,\gamma } \right)} \right].
\end{equation}So, considering constraint~\eqref{first_opt3}, we can write,
\begin{equation}\label{mu_proof}
\mu _a \left( {t,i} \right) \ge \frac{{{R_{\rm 0}}  \lambda _u \left(t, i \right)}}{{B\lambda _r \left( i \right)\left[ {\log _2 (1 + \gamma ) + \gamma ^{\frac{2}{\alpha }} \mathcal{A}\left( {\alpha ,\gamma } \right)} \right]}},
\end{equation}which establishes the minimum RRH activity factor as a function of $\lambda_u(t,i)$ to satisfy the per-user data-rate constraint.
\end{proof}

\subsection{Optimal Transmission Power}
Given a fixed active RRH density, we can minimize the transmit power of the active RRHs so as to achieve a certain coverage and outage probability. Since in our solution the active RRH density of different clusters changes over time based on the traffic demand, we need to also dynamically optimize the transmit power accordingly. This can further decrease the power consumption of the system. For instance, when the density of active RRHs becomes higher, each RRH has only a small coverage area and users can be in coverage even with a lower transmission power.

\begin{lem}\label{lemma_1}
The minimum transmission power for which the constraint ${{P_{{\mathop{\rm cov}} }} \ge \varepsilon P_{{\mathop{\rm cov}} }^{\infty}}$ is met is given by,
\begin{equation}\label{P_1}
P_c^*(t,i) = \frac{{{L_1}}}{{{{\left[ {{\mu _a}(t,i){\lambda _r}(i)} \right]}^{\frac{\alpha }{2}}}}},
\end{equation}
where
\begin{subequations}
\label{L1}
\begin{align}
&{L_1} = \frac{{\gamma {\sigma ^2}\Gamma \left( {\frac{\alpha }{2} + 1} \right)}}{{{\pi ^{\frac{\alpha }{2}}}{{\left[ {1 + \Upsilon (\gamma ,\alpha )} \right]}^{\frac{\alpha }{2}}}\left( {1 - \varepsilon } \right)}}, \label{eq:L1_1}\\
&\Upsilon (\gamma ,\alpha ) = {\gamma ^{\frac{2}{\alpha }}}\int_0^\infty  {\frac{1}{{1 + {z^{\frac{\alpha }{2}}}}}} dz. \label{eq:L1_2}
\end{align}
\end{subequations}
\end{lem}

\begin{proof}
The coverage probability in the $i^{\rm th}$ cluster is given as~\cite{andrews2011tractable},
\begin{multline}\label{P_cov_p}
P_{{\mathop{\rm cov}} } \left( {\alpha ,\gamma ,\mu _a } \right) = \pi \mu _a \left( {t,i} \right)\lambda _r \left( i \right) \\ \times \int_0^\infty  {e^{ - \pi \mu _a \left( {t,i} \right)\lambda _r \left( i \right)\left( {1 + \Upsilon \left( {\alpha ,\gamma } \right)} \right) - \gamma \sigma ^2 v^{\alpha /2} P^{ - 1} } dv}.
\end{multline}Now, by using the substitution $\gamma \sigma ^2 P^{ - 1}  \to s$ in~\eqref{P_cov_p} and the approximation $e^{ - sv^{\alpha /2} }  \approx \left( {1 - sv^{\alpha /2} } \right)$ (in the case of low-noise regimes, i.e., $\sigma_n^2  \to 0$), we can write,
\begin{multline}\label{P_cov_apr}
P_{{\mathop{\rm cov}} } \left( {\alpha ,\gamma ,\mu _a } \right) \approx \\ P^\infty  \left( {1 - \frac{{\gamma \sigma ^2 \Gamma \left( {\frac{\alpha }{2} + 1} \right)}}{{P\left[ {\pi \mu _a \left( {t,i} \right)\lambda _r \left( i \right)\left( {1 + \Upsilon \left( {\alpha ,\gamma } \right)} \right)} \right]^{\frac{\alpha }{2}} }}} \right),
\end{multline}where $P^\infty$ is the coverage probability without noise~\cite{andrews2011tractable}, i.e.,
\begin{equation}\label{P_infinity}
P^\infty   = \left( {1 + \Upsilon \left( {\alpha ,\gamma } \right)} \right)^{ - 1},
\end{equation}and $\Gamma \left( x \right) = \int_0^\infty  {t^{x - 1} e^{ - t} dt}$ is the standard gamma function. The minimum transmission power $P_c^*(t,i)$ that satisfies the coverage constraints is obtained by combining~\eqref{P_cov_apr} and \eqref{first_opt2}.
\end{proof}

\subsection{Optimal Size of VM}\label{VBS_size}
We recast the power consumption of the $i^{th}$ VBS-Cluster,
\begin{equation}\label{VBS_Power_recast}
P_{vm}^i = \Delta \left( {t,i} \right){P_{\max }}u\left( {t,i} \right)(1 - \beta ) + \beta \Delta \left( {t,i} \right){P_{\max }}+P_{olt},
\end{equation}where, for a given workload, $u(t,i)$ is inversely proportional to $\Delta(t,i)$~\cite{trivedi1980optimal}. So, to minimize the power consumption of the VM, we need to minimize the size of the VM (CPU cores) such that the network requirements are met. The workload in the VBS-Cluster depends on the LTE MCS index, number of PRBs, and the channel bandwidth. Moreover, according to~\cite{chanclou2013optical}, the Round Trip Time~(RTT) between RRH and VBS pool cannot exceed  $400~\rm{\mu s}$. Since the total delay budget in LTE is considered as $3~\rm{ms}$, this leaves the VBS-Cluster with only about $2.6~\rm{ms}$ for signal processing. For this matter, we consider a modified model of the processing time presented in~\cite{alyafawi2015critical}, which is given by,
\begin{equation}\label{T_frame}
T_{{\rm{fr}}}  = \frac{{M\upsilon }}{\Delta (t,i) } = \frac{{M\upsilon }}{{N_c(t,i)\omega }},
\end{equation}where $T_{{\rm{fr}}}$ is the processing time and is measured in $\rm{\mu s}$, $M$ is the number of PRBs, $\upsilon$ is a MSC-dependent constant, $N_c(t,i)$ is the number of dedicated CPU cores to the $i^{th}$ VBS-Cluster, and $\omega$ is the CPU speed measured in $\rm{GHz}$. Hence, the minimum number of required CPU cores to meet the frame deadline is given by,
\begin{equation}\label{N_c}
N_c^*(t,i)  = \left\lceil {\frac{{M\upsilon }}{{T_{{\rm{dl}}} \omega }}} \right\rceil,
\end{equation}where $T_{{\rm{dl}}}$ is the frame deadline.

\begin{table}[ht!]
  \caption{Simulation Parameters.}\label{tab:assumptions}
\centering
\resizebox{8.5cm}{!} {
\begin{tabular}{|c|c|} \hline \hline
\multicolumn{1}{c|}{\textbf{Parameters}} & \multicolumn{1}{c}{\textbf{Mode/Value}}  \\ \hline \hline
\multicolumn{1}{c|}{\shortstack{\\ Cellular Layout}} & \multicolumn{1}{c}{Homogeneous Poisson Point Process}  \\ \hline
\multicolumn{1}{c|}{\shortstack{\\ Channel Model}} & \multicolumn{1}{c}{Path Loss and Shadowing}  \\ \hline
\multicolumn{1}{c|}{\shortstack{\\ Channel Bandwidth}} & \multicolumn{1}{c}{$20~\rm{MHz}$}  \\ \hline
\multicolumn{1}{c|}{\shortstack{\\ Number of Antennas $(N_{TX},N_{RX})$}} & \multicolumn{1}{c}{$(1,1)$}  \\ \hline
\multicolumn{1}{c|}{\shortstack{\\ OLT power consumption ($P_{olt}$)}} & \multicolumn{1}{c}{$20~\rm{W}$}  \\ \hline
\multicolumn{1}{c|}{\shortstack{\\ ONU Power Consumption in Active Mode ($P_{tl}^{a}$)}} & \multicolumn{1}{c}{$4~\rm{W}$}  \\ \hline
\multicolumn{1}{c|}{\shortstack{\\ ONU Power Consumption in Sleep Mode ($P_{tl}^{s}$)}} & \multicolumn{1}{c}{$0.5~\rm{W}$}  \\ \hline
\multicolumn{1}{c|}{\shortstack{\\ RRH Circuit Power Consumption in Active Mode ($P_{rrh}^{a}$)}} & \multicolumn{1}{c}{$12.4~\rm{W}$}  \\ \hline
\multicolumn{1}{c|}{\shortstack{\\ RRH power consumption in sleep mode ($P_{rrh}^{s}$)}} & \multicolumn{1}{c}{$3.5~\rm{W}$}  \\ \hline
\multicolumn{1}{c|}{\shortstack{\\ Maximum power consumed per each CPU core ($P_{max}$)}} & \multicolumn{1}{c}{$72~\rm{W}$}  \\ \hline
\multicolumn{1}{c|}{\shortstack{\\ Power Amplifier Efficiency ($\eta$)}} & \multicolumn{1}{c}{$0.32$}  \\ \hline
\multicolumn{1}{c|}{\shortstack{\\ MSC Dependent Constant ($\alpha$)}} & \multicolumn{1}{c}{$117.4$}  \\ \hline
\multicolumn{1}{c|}{\shortstack{\\ Fraction of Power Consumed by Idle VBS ($\beta$)}} & \multicolumn{1}{c}{$0.7$}  \\ \hline
\multicolumn{1}{c|}{\shortstack{\\ Fraction of Minimum Coverage Probability ($\varepsilon$)}} & \multicolumn{1}{c}{$0.75$}  \\ \hline
\multicolumn{1}{c|}{\shortstack{\\ Minimum Data Rate ($R_0$)}} & \multicolumn{1}{c}{$200~\rm{Kbps}$}  \\ \hline
\end{tabular}
}
\end{table}

\section{Performance Evaluation}\label{sec:simulation}
In this section, we provide a range of simulations and real-time emulations to evaluate the performance of our solution. In the simulations, we consider a cellular network where the RRHs and the users are distributed according to two independent homogeneous PPPs. Table~\ref{tab:assumptions} lists the stimulation parameters used during our experiments.
\begin{figure*}[t]
 \centering
 \resizebox{\linewidth}{!}{
 \begin{tabular}{cccc}
\includegraphics[width=0.35\textwidth,height=2in]{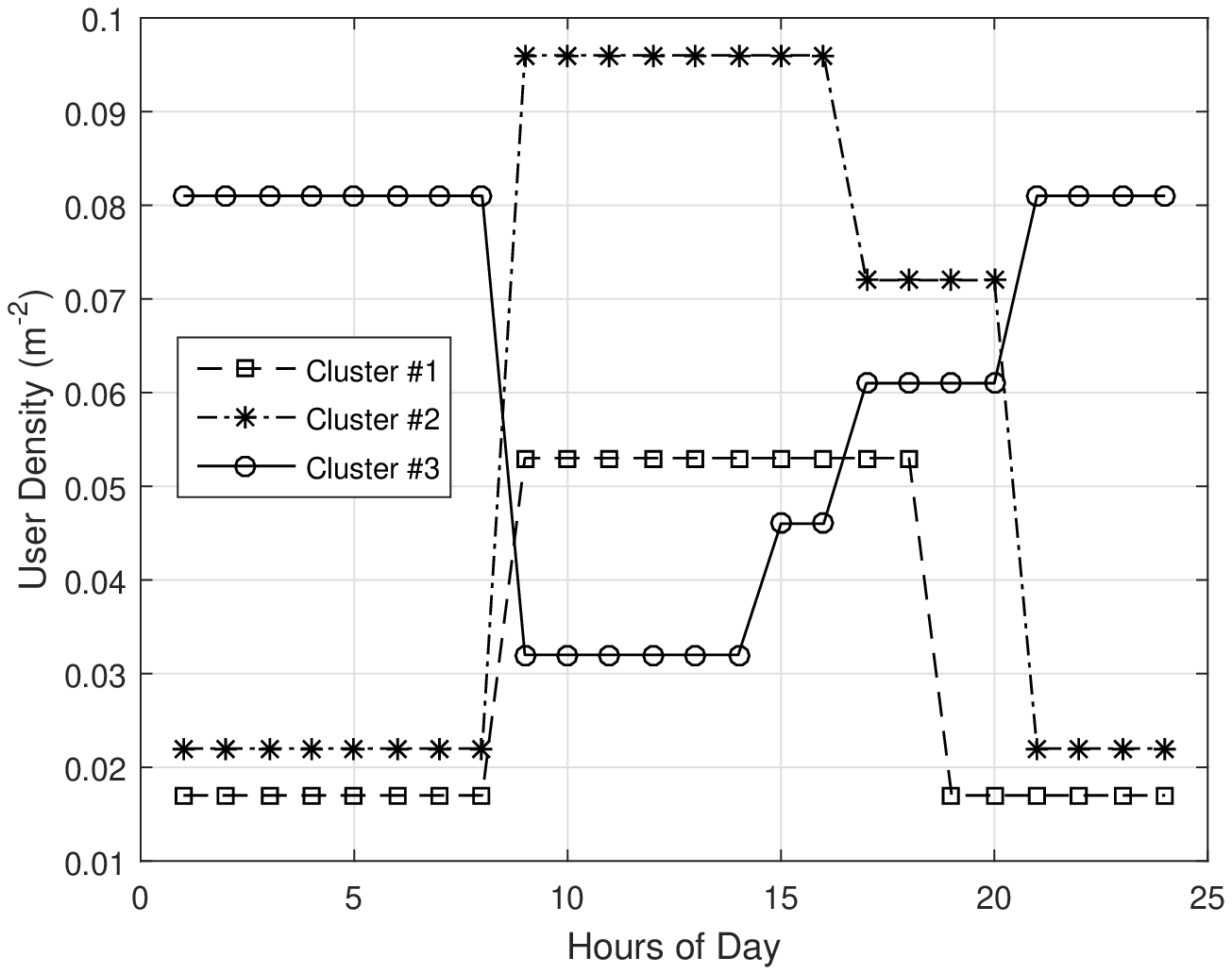} &
\includegraphics[width=0.35\textwidth,height=2in]{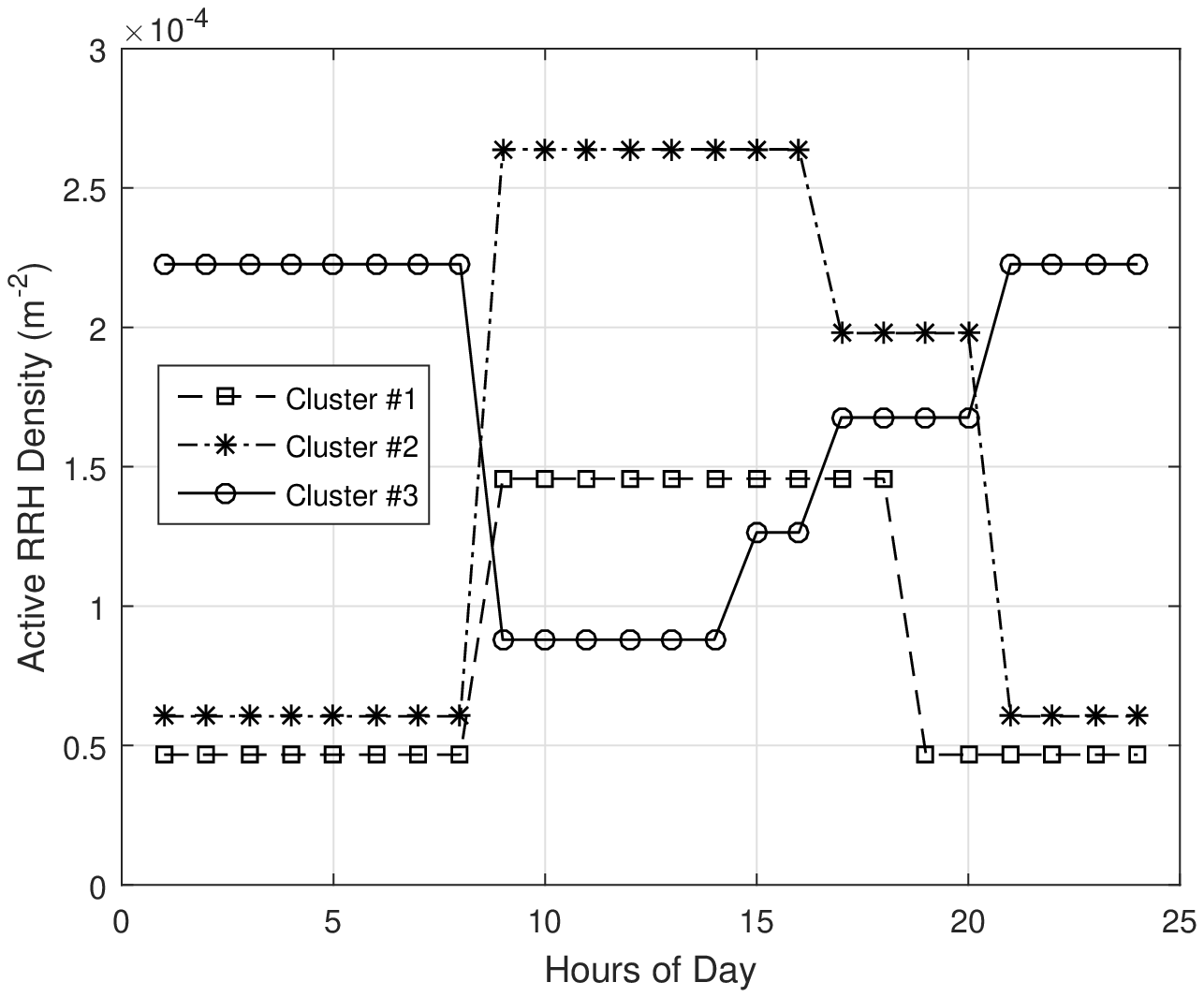} &
\includegraphics[width=0.35\textwidth,height=2in]{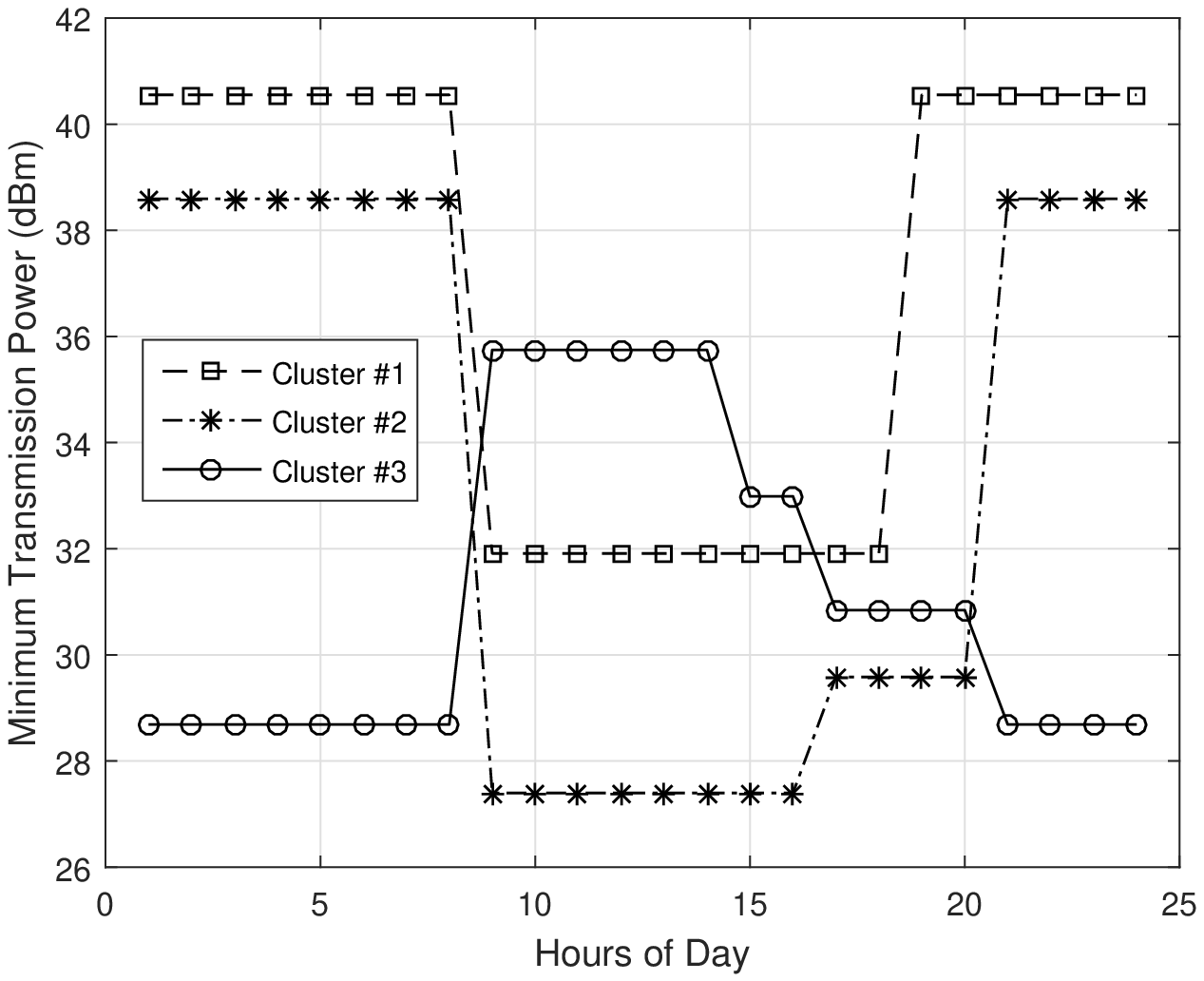} &
\includegraphics[width=0.35\textwidth,height=2in]{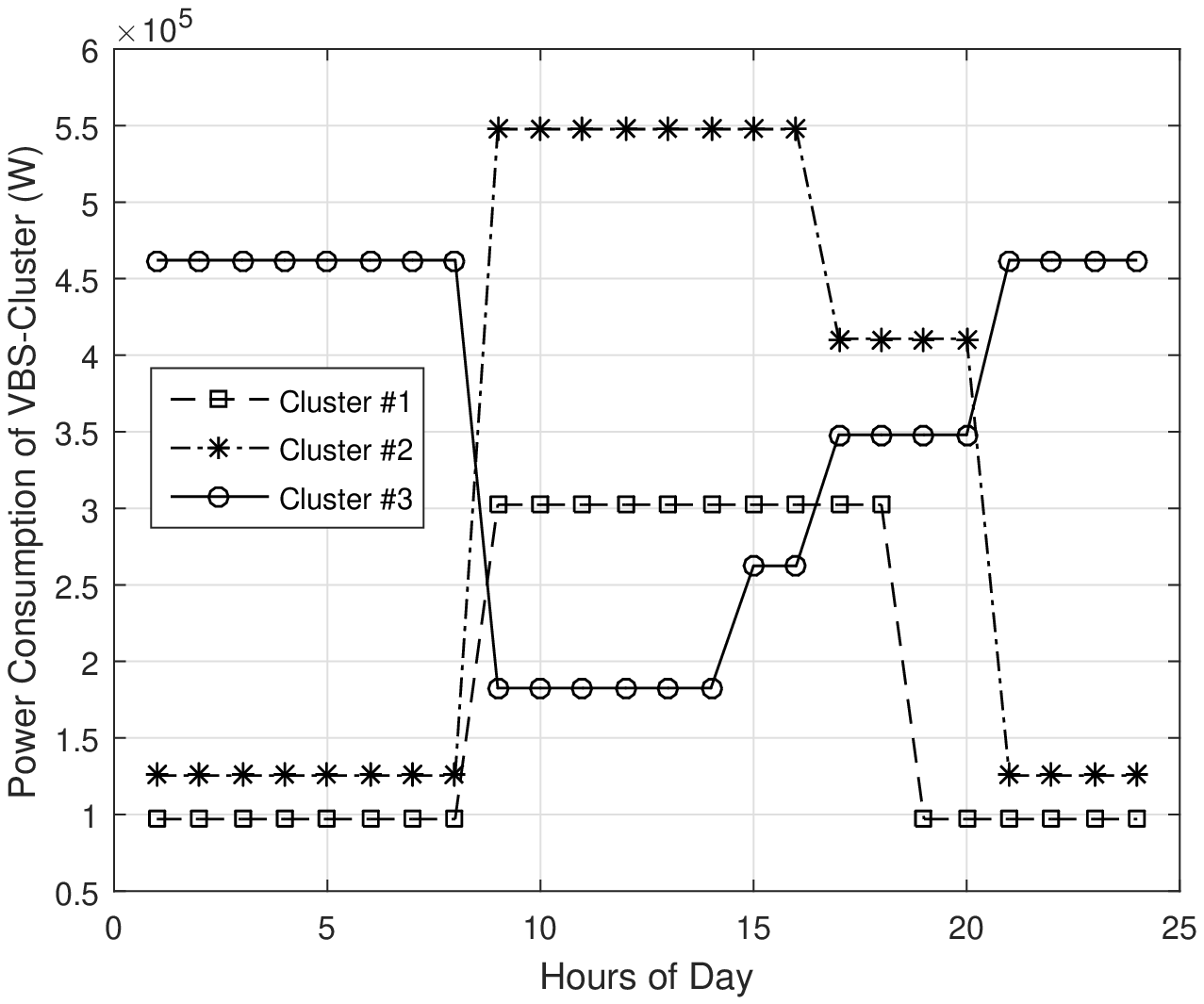} \\
\small (a) & \small(b) & \small(c) & \small(d)
\end{tabular}}
\caption{(a) Traffic fluctuation on a typical operational day for three different clusters; (b) Fluctuation of active RRH density to meet the user fluctuation for different times of the day; (c) Fluctuation of minimum transmission power to reduce the power consumption while guaranteeing a predefined Quality of Service~(QoS); (d) Power consumption of different VBS-Clusters (CPU speed = $3.3~\rm{GHz}$).}\label{fig:sim3}
\end{figure*}
%
In order to show the performance of our dynamic solution, we simulate the traffic fluctuation on a typical operational day and show how \emph{Elastic-Net} dynamically adapts the RRH density, transmission power, and size of VBS-Clusters to minimize the power consumption while at the same time meeting the network constraints. As shown in Fig.~\ref{fig:cran_example}, we consider a network with $3$ clusters where each cluster covers an area of $25$~\rm{km}$^2$ and has its own traffic and user density fluctuation.

As shown in Fig.~\ref{fig:sim3}(a), clusters~\#1 and \#2 (corresponding to the downtown and entertainment areas, respectively) have lower user density on early morning and late night, while cluster~\#3 (residential area) has higher user density at those times. Figure~\ref{fig:sim3}(b) illustrates the minimum active RRH density adaptation required to serve the corresponding user density fluctuation in different clusters. As expected, the RRH density fluctuation corresponds to the user density fluctuation. This means that for the high traffic demand times we need a higher number of active RRHs and smaller cells. The time varying transmission power for different clusters is shown in  Fig.~\ref{fig:sim3}(c). It is clear that for the peak traffic times we need a lower transmission power than in low traffic times. This is because in the peak traffic times we have higher active RRH density and, consequently, the coverage area of each RRH becomes smaller. Figure~\ref{fig:sim3}(d) also shows the power consumption of VBS-Clusters in the VBS pool.

In Figs.~\ref{fig:sim4}(a-c), we compare the traditional static provisioning against Elastic-Net. As shown in~(a), (b), and (c), depending on the traffic fluctuation there is a noticeable difference between the power consumption of Elastic-Net and Static-Net. For instance, in cluster~\#1 and for low traffic times (7~PM--8~AM) we have $48.6\%$ decrease in power consumption by Elastic-Net. However, for the peak traffic times (8~AM--7~PM), we have only $7.4\%$ decrease in power consumption. This confirms our statements in Sect.~\ref{sec:introduction} that a small-cell deployment is efficient when the capacity demand or user density is high, while it becomes less so when the traffic demand is low.

\begin{figure*}[t]
 \centering
 \resizebox{\linewidth}{!}{
 \begin{tabular}{ccc}
\includegraphics[width=0.37\textwidth, height=1.9in]{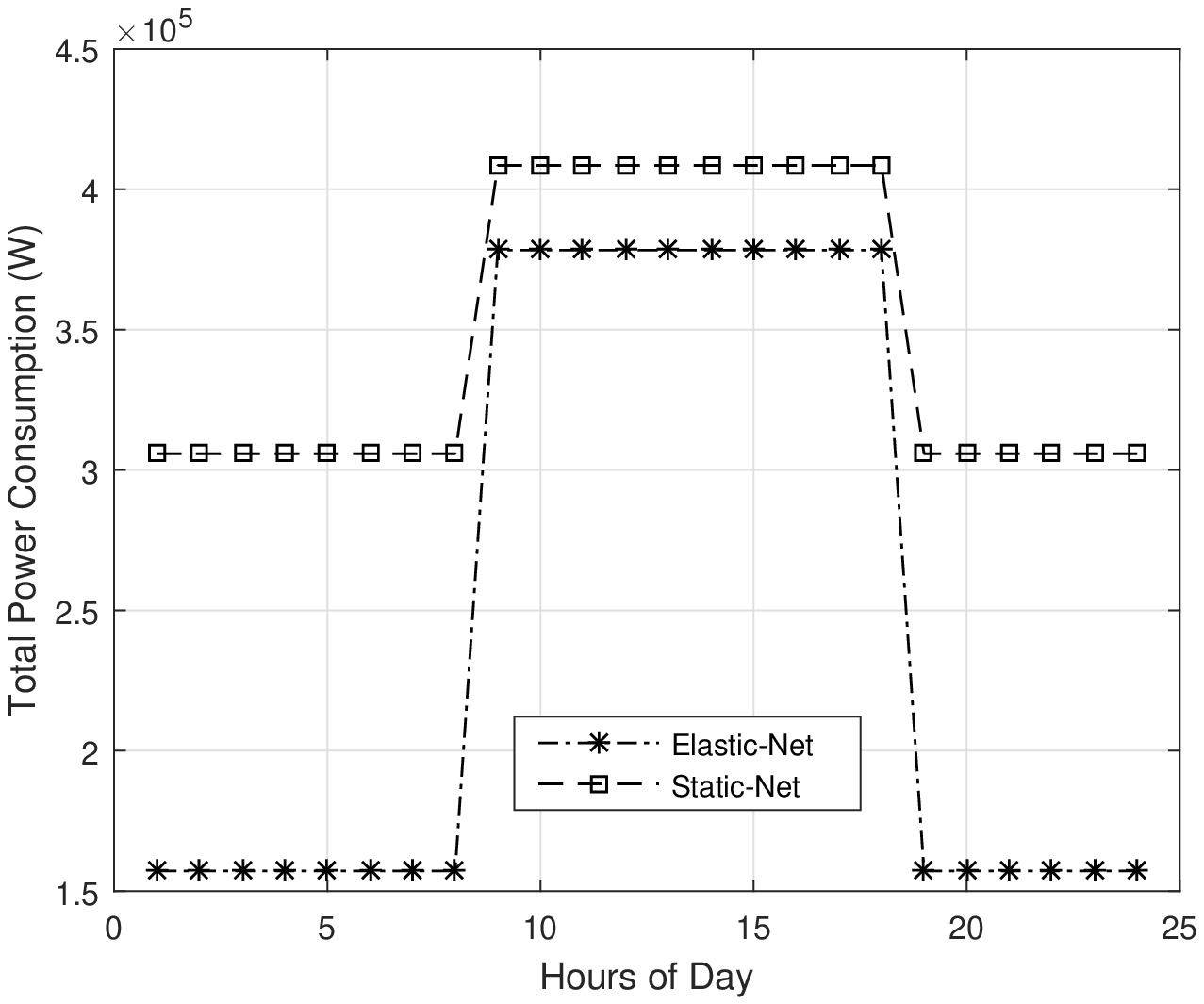} &
\includegraphics[width=0.37\textwidth, height=1.9in]{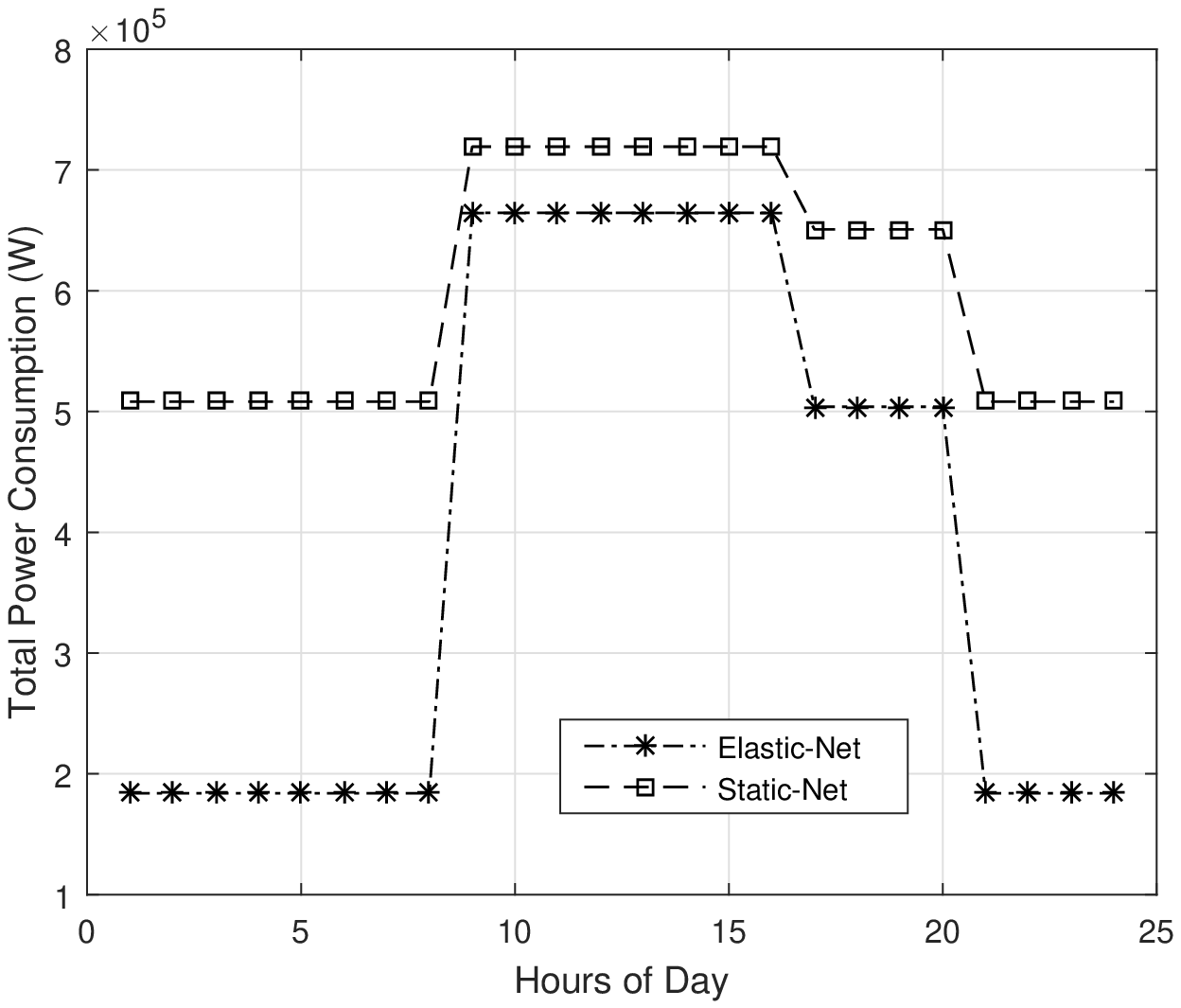} &
\includegraphics[width=0.37\textwidth, height=1.9in]{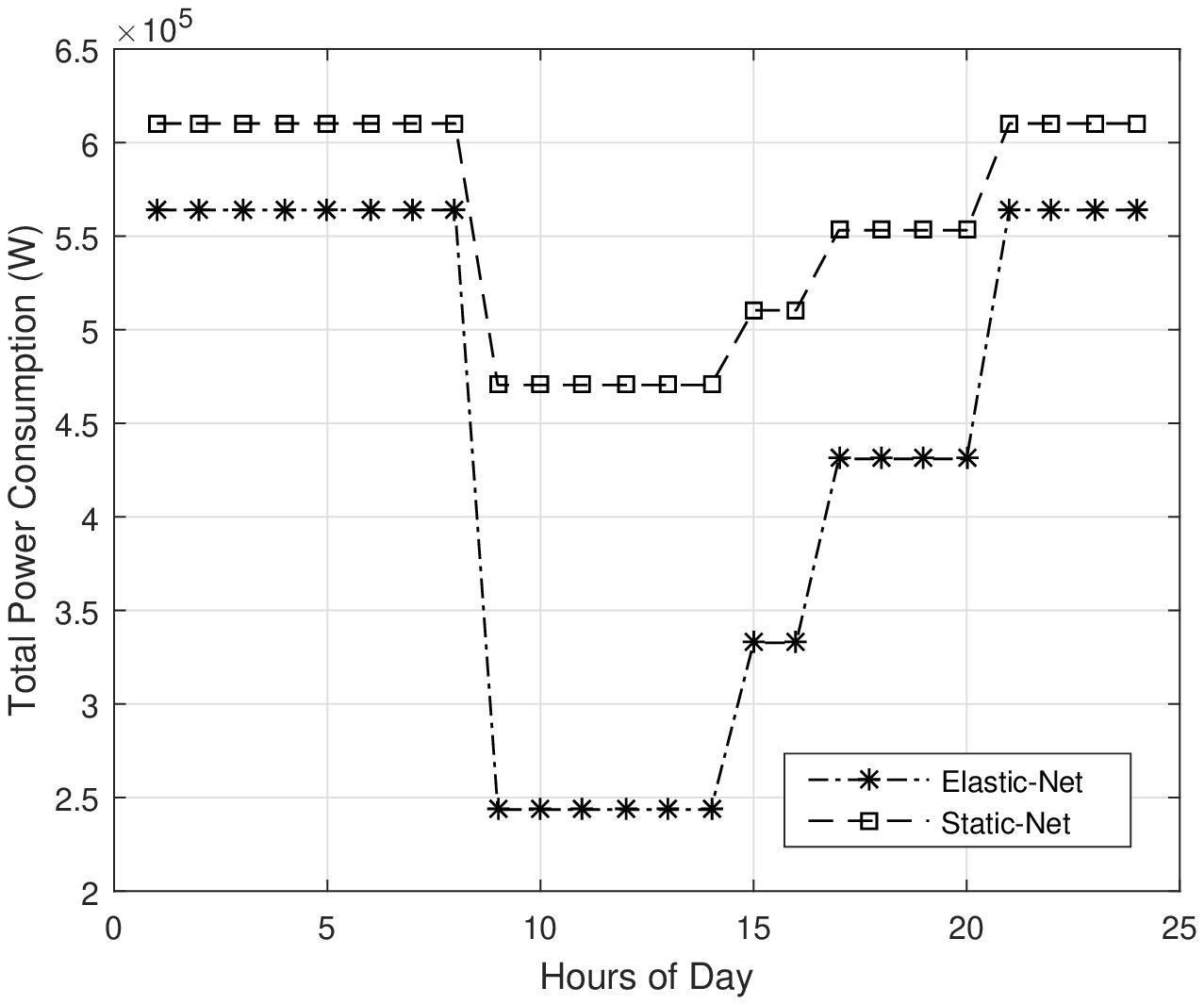} \\
\small (a) & \small(b) & \small(c)
\end{tabular}}
\caption{Comparison of power consumption between Elastic-Net~(in C-RAN) and Static-Net~(in D-RAN) for (a) Cluster~\#1, (b) Cluster~\#2, and (c) Cluster~\#3.}\label{fig:sim4}
\end{figure*}

\balance

\section{Conclusion}\label{sec:conclusion}
In the context of C-RAN---a new centralized paradigm for wireless cellular networks in which the Base Stations~(BSs) are physically decoupled into Virtual Base Stations~(VBSs) and Remote Radio Heads~(RRHs)---we proposed a novel demand-aware reconfigurable solution, named~\emph{Elastic-Net}, to minimize the network power consumption and to adapt to the fluctuations in per-user capacity demand. We divided the covered region of the cellular network into clusters based on active traffic, and dynamically adjusted RRH density, transmission power, and size of the Virtual Machine~(VM) holding the VBS so that the network power consumption is minimized and the network constraints are met. Simulation results corroborated our analysis and confirmed the benefits of our solution.

\textbf{Acknowledgment:} This work was supported in part by the National Science Foundation Grant No.~CNS-1319945.

\bibliographystyle{IEEEtran}\small
\bibliography{infocom}

\end{document}